\documentclass[10pt]{article}

\usepackage[margin=1in]{geometry}
\usepackage{amsmath}
\usepackage{amssymb}
\usepackage{mathtools}
\usepackage{amsthm}
\usepackage{braket}
\usepackage{tensor}
\usepackage{graphicx}
\usepackage{float}
\usepackage{tikz}
\usepackage{threeparttable}
\usepackage{enumitem}
\usepackage{indentfirst}
\usepackage{appendix}
\usepackage{algorithm}
\usepackage{algpseudocode}
\usepackage{subcaption}
\usepackage{makecell}
\usepackage{multirow}
\usepackage[square, numbers]{natbib}
\usepackage{authblk}
\usepackage{hyperref}
\newenvironment{acknowledgments}{\section*{Acknowledgments}}{}

\theoremstyle{definition}
\newtheorem{definition}{Definition}
\theoremstyle{plain}

\newtheorem{theorem}{Theorem}

\newtheorem{lemma}{Lemma}

\newtheorem{problem}{Problem}
\newtheorem{example}{Example}

\title{Automated Construction and Verification of Unextendible Product Bases}

\author[1]{Zicheng Han\thanks{ \href{mailto:hzcqj2020@mail.ustc.edu.cn}{hzcqj2020@mail.ustc.edu.cn}}}
\author[2]{Wanchen Zhang\thanks{ \href{mailto:wanchenz@mail.ustc.edu.cn}{wanchenz@mail.ustc.edu.cn}}}
\author[3]{Fei Shi\thanks{\href{mailto:shif26@mail.sysu.edu.cn}{shif26@mail.sysu.edu.cn}}}
\author[1,2]{Xiande Zhang \thanks{Corresponding author: \href{mailto:drzhangx@ustc.edu.cn}{drzhangx@ustc.edu.cn}}}

\affil[1]{School of Mathematical Sciences, University of Science and Technology of China, Hefei 230026, China}
\affil[2]{Hefei National Laboratory, University of Science and Technology of China, Hefei 230088, China}
\affil[3]{Institute of Quantum Computing and Software, School
of Computer Science and Engineering, Sun Yat-sen University, Guangzhou 510006, China}
\date{} 
\begin{document}
\maketitle
\begin{abstract}

Unextendible product bases (UPBs) are important structures in quantum information theory, with applications to completely entangled subspaces, bound entanglement, and local indistinguishability. Since many properties and applications of UPBs are closely related to their cardinalities, one of the central problems in the study of UPBs is to determine whether UPBs of prescribed sizes exist in a given multipartite system. In this paper, we introduce a SAT-assisted framework based on decompositions of the \(N\)-dimensional hypercube.
We define \(O_N\)-tile decompositions and prove a tile-to-UPB theorem: every \(O_N\)-tile decomposition induces a UPB through a construction based on tile-wise Fourier product bases and a global stopper state. We then encode the search for such decompositions as a Boolean satisfiability (SAT) problem and use SAT solvers to generate explicit instances. In terms of verification, we also implement a UPB verification algorithm based on local orthogonality graphs and unsaturated subspaces. The algorithm can be used to determine whether an arbitrary finite set of product states forms a UPB. Using this framework, we obtain UPBs of several sizes in some tripartite and quadripartite systems, including sizes \(13,14,\ldots,23\) in \(\mathbb C^3\otimes\mathbb C^3\otimes\mathbb C^3\). Moreover, the small-dimensional instances obtained here can serve as seed UPBs for recursive constructions, leading to further examples in larger multipartite systems.

\end{abstract}

\section{Introduction}
Entanglement is a central resource in quantum information theory, underlying fundamental protocols such as quantum teleportation and quantum key distribution~\cite{horodecki2009quantum,bennett1993teleporting,ekert1991quantum}. Determining whether a general multipartite quantum state is separable or entangled is computationally difficult~\cite{terhal2002detecting,gurvits2004classical,gharibian2010strong}. This difficulty motivates the systematic construction of quantum states and subspaces whose entanglement properties are known a priori, both for theoretical studies and for data-driven approaches to entanglement detection~\cite{deng2017quantum,lu2018separability,harney2020entanglement,bai2022unsupervised,zhang2023entanglement,li2024entanglement,huang2025direct}.

A systematic source of such certified subspaces is provided by unextendible product bases (UPBs)~\cite{bennett1999unextendible,divincenzo2003unextendible}. A UPB is a finite set of mutually orthogonal product states such that the orthogonal complement of its span contains no nonzero product state. Hence the orthogonal complement of a UPB is a completely entangled subspace (CES), in which every nonzero pure state is entangled~\cite{parthasarathy2004maximal,demianowicz2018unextendible,demianowicz2024completely}. This observation gives a standard method for constructing bound entangled states~\cite{bennett1999unextendible,horodecki1998mixed}, namely entangled mixed states from which no pure entanglement can be distilled by local operations and classical communication (LOCC). UPBs also play an important role in related questions such as quantum nonlocality without entanglement and Bell inequalities with no quantum violation~\cite{bennett1999quantum,augusiak2011bell,augusiak2012tight}.

Motivated by these applications, the construction and characterization of UPBs have been studied extensively for more than two decades~\cite{bennett1999unextendible,divincenzo2003unextendible,alon2001unextendible,feng2006unextendible,chen2015minimum,agrawal2019genuinely}. Since UPBs are constrained simultaneously by local tensor-product structure and global unextendibility, their construction in multipartite Hilbert spaces with large local dimensions is naturally combinatorial. Two approaches are particularly relevant to the present work: hypercube-decomposition methods~\cite{bennett1999quantum,shi2020unextendible,you2023unextendible,che2022constructing,he2024strong,shi2022strongly}, which provide a geometric route from finite product grids to explicit UPBs, and orthogonality-graph methods~\cite{alon2001unextendible,feng2006unextendible,chen2015minimum,johnston2014structure,shi2023graph}, which encode local orthogonality relations and support structural characterization and verification.

\textbf{Construction of UPBs via hypercube decomposition:} 
As a standard approach for UPB construction, one performs a specific decomposition on the $N$-dimensional hypercube ($N$-cube) $\mathcal{C} = \mathbb{Z}_{d_1} \times\cdots\times \mathbb{Z}_{d_N}$ ($\mathbb{Z}_n:= \{0, 1, \dots, n-1\}$), and subsequently maps this decomposition directly into a UPB in the Hilbert space $\mathcal{H} = \bigotimes_{i=1}^N \mathbb{C}^{d_i}$. This geometric approach was pioneered for bipartite systems ($N=2$) by Bennett \textit{et al.}~\cite{bennett1999quantum}, who constructed the seminal ``tiles'' UPB to demonstrate quantum nonlocality without entanglement. This intuition was subsequently formalized by Shi \textit{et al.}~\cite{shi2020unextendible} through the $U$-tile structure, establishing a rigorous correspondence between $2$-cube   decompositions and UPBs. The algorithmic side of this approach was further developed by You \textit{et al.}~\cite{you2023unextendible}, who introduced \(O\)-tiles for $2$-cube decompositions.

Although hypercube decomposition methods have been extended to multipartite systems with \(N\ge 3\), finding valid decompositions becomes increasingly difficult as the number of parties and local dimensions grow. Existing multipartite constructions still often rely on manually designed seed decompositions and recursive extension procedures~\cite{che2022constructing,zhen2024unextendible,shi2022strongly,shi2021strong,he2024strong,agrawal2019genuinely}. In particular, a common strategy is to start from decompositions of $\mathcal C=\mathbb Z_3\times\mathbb Z_3\times\mathbb Z_3$ and use them as seeds for systems with larger local dimensions~\cite{che2022constructing,he2024strong,shi2022strongly,shi2021strong,agrawal2019genuinely}. Thus, the systematic generation of new seed decompositions with diverse sizes and structures remains a natural computational challenge.

\textbf{Characterization and verification of UPBs via orthogonality graphs:} Pioneered by Alon and Lovász~\cite{alon2001unextendible}, the orthogonality graph representation was first instrumental in deriving universal lower bounds for the cardinality of UPBs. Building upon this framework, Feng~\cite{feng2006unextendible} utilized the 1-factorization of complete graphs to determine the minimum size of UPBs in certain Hilbert spaces. Subsequently, this graph-theoretic perspective was expanded by Chen and Johnston~\cite{chen2015minimum} to address the determination of minimum UPB cardinalities in bipartite case and some multipartite cases, verifying the tightness of the Alon-Lovász bound. Beyond cardinality constraints, regarding structural properties, Johnston~\cite{johnston2014structure} introduced computational search techniques based on the orthogonality graph representation to construct and classify specific classes of qubit UPBs. More recently, Shi \textit{et al.}~\cite{shi2023graph} characterized the size and properties of genuinely unextendible product bases (GUPBs) via the orthogonality graph representation of UPBs. This line of work provides the basis for verification methods based on local orthogonality graphs and unsaturated subspaces.

Indeed, prior studies have explored computer-aided methods for constructing UPBs~\cite{you2023unextendible,johnston2014structure,cheng2024construction}. However, these methods are typically restricted to bipartite systems or qubit cases, while multipartite systems with higher local dimensions still require substantial manual input. In this work, we develop a systematic SAT-assisted construction framework based on \(N\)-cube decompositions. The main contributions are as follows:

\begin{enumerate}
\item \textbf{An explicit $N$-dimensional tile-to-UPB theorem.} We formulate an $O_N$-tile condition on
$N$-cube \(\mathcal{C}=\mathbb Z_{d_1}\times\cdots\times\mathbb Z_{d_N}\), and prove an \(N\)-partite tile-to-UPB theorem that generalizes the bipartite tile-structure criterion~\cite{you2023unextendible}. That is, every \(O_N\)-tile decomposition induces a UPB through a construction based on tile-wise Fourier product bases and a global stopper state. This gives a structural bridge from hypercube decompositions to UPBs.

\item \textbf{SAT encoding of \(O_N\)-tile decompositions.} For the hypercube
\(\mathcal C=\mathbb Z_{d_1}\times\cdots\times\mathbb Z_{d_N}\), we encode the search for \(O_N\)-tile decompositions as a Boolean satisfiability problem. The encoding includes admissibility, covering, non-overlap, non-combinability, and cardinality constraints.

\item \textbf{A general verification algorithm.} We implement an exact verification algorithm based on local orthogonality graphs and unsaturated subspaces. The algorithm applies to arbitrary finite sets of product states and is used here to certify the symbolic product bases generated from the SAT-found \(O_N\)-tile decompositions. 

\item \textbf{Explicit UPB instances.} Using the SAT search, we construct UPBs of several sizes in selected tripartite and quadripartite systems. In particular, we obtain UPBs of size \(13,14,\ldots,23\) in \(\mathbb C^3\otimes\mathbb C^3\otimes\mathbb C^3\).  Previously, only sizes $7$ and $19$ are known~\cite{divincenzo2003unextendible,alon2001unextendible,agrawal2019genuinely}.
\end{enumerate}

Collectively, these ingredients provide an automated pipeline from combinatorial hypercube decompositions to explicit UPBs, together with independently checkable computational certificates. The verification routine is also benchmarked against the existing \textit{IsUPB} function in QETLAB~\cite{qetlab} on representative instances. Finally, the results supply seed examples for recursive constructions and further studies of completely entangled subspaces and bound entangled states.

The remainder of this paper is organized as follows. Section~2 recalls UPBs and develops the hypercube-decomposition construction, with Theorem~\ref{thm:tile-to-UPB} as the main tile-to-UPB theorem. Section~3 gives the SAT formulation used to search for \(O_N\)-tile decompositions. Section~4 presents a general orthogonality-graph and local-subspace verification algorithm for UPBs, which is applied here to certify the generated symbolic product bases. Section~5 reports the constructed UPB instances and discusses their consequences and limitations.

\section{Unextendible Product Bases and the Tile-to-UPB Construction}
In this section, we introduce the preliminaries and basic facts. For simplicity, 
denote $\mathbb{Z}_n := \{0, 1, \dots, n-1\}$.

\subsection{Unextendible product bases}
 We consider an $N$-partite Hilbert space $\mathcal{H} = \bigotimes_{i=1}^N \mathbb{C}^{d_i}$. For convenience, we assume the local dimensions are ordered such that $d_1 \leq d_2 \leq \cdots \leq d_N$. Unless otherwise stated, the states and operators discussed are unnormalized.

Any $N$-partite pure state in this space can be expressed as:

$$
\ket{\Psi} = \sum_{i_1,i_2,\ldots,i_N} c_{i_1,i_2,\ldots,i_N} \ket{i_1} \ket{i_2}\cdots\ket{i_N},\footnotemark
$$
\footnotetext{For brevity, the tensor product symbol $\otimes$ is omitted here and throughout the paper.}
where $\{ \ket{i_j} \}_{i_j=0}^{d_j-1}$ represents the computational basis of the $j$-th subsystem $\mathbb{C}^{d_j}$.

We classify sets of orthogonal product states based on the subspace they span. A complete orthogonal product basis (COPB) is a set of orthogonal product states that spans the entire space $\mathcal{H}$. An orthogonal product set is called a UPB if its orthogonal complement contains no nonzero product state. It is nontrivial if its cardinality is strictly smaller than $\dim \mathcal H$.
The formal definition is as follows:

\begin{definition}
Let $\mathcal{U} = \{ \ket{\varphi^{(i)}} \}_{i=1}^k \subset \mathcal{H} = \bigotimes_{i=1}^N \mathbb{C}^{d_i}$ be a set of orthogonal product states, where each state is of the form $\ket{\varphi^{(i)}} = \ket{\varphi_1^{(i)}}  \ket{\varphi_2^{(i)}}  \cdots \ket{\varphi_N^{(i)}}$. The set $\mathcal{U}$ is called a \textbf{UPB} if the orthogonal complement of the subspace spanned by $\mathcal{U}$ contains no nonzero product state.
\end{definition}

Note that a UPB of size $\prod_{i=1}^N d_i$ is trivial, since a COPB always meets the criterion. Consequently, we restrict our attention to nontrivial UPBs throughout this work.
Example~\ref{shift_upb} illustrates the canonical ``Shifts'' UPB~\footnote{It is named the ``Shifts'' UPB because the first three orthogonal product states are constructed by cyclically shifting a basic set of local single-qubit states across different subsystems.} in a three-qubit system $\mathbb{C}^2 \otimes \mathbb{C}^2 \otimes \mathbb{C}^2$, which was originally proposed by Bennett \textit{et al.}~\cite{bennett1999unextendible}.
\begin{example}[The ``Shifts'' UPB\cite{bennett1999unextendible}]\label{shift_upb}
The ``Shifts'' UPB consists of the following four  product states:
\begin{equation}
\begin{aligned}
\ket{\varphi^{(1)}} &= \ket{0}\ket{1}(\ket{0}-\ket{1}), \\
\ket{\varphi^{(2)}} &= \ket{1}(\ket{0}-\ket{1})\ket{0}, \\
\ket{\varphi^{(3)}} &= (\ket{0}-\ket{1})\ket{0}\ket{1}, \\
\ket{S} &= (\ket{0}+\ket{1})(\ket{0}+\ket{1})(\ket{0}+\ket{1}).
\end{aligned}
\end{equation}
It is easy to verify that these states are mutually orthogonal and that there is no nonzero product state in the orthogonal complement of the subspace spanned by $\{ \ket{\varphi^{(i)}} \}_{i=1}^3\cup \{\ket{S}\}$.
\end{example}

A fundamental and longstanding question regarding UPBs is:

\begin{problem}\label{Pro_size}
    For which sizes do UPBs exist in the $N$-partite Hilbert space $\mathcal{H} = \bigotimes_{i=1}^N \mathbb{C}^{d_i}$?
\end{problem} 
Extensive research has investigated upper and lower bounds for UPB sizes~\cite{divincenzo2003unextendible,alon2001unextendible,feng2006unextendible,chen2015minimum,johnston2014structure,johnston2013minimum}. In particular,  Alon and Lov\'asz ~\cite{alon2001unextendible} gave the universal lower bound
\[
    |\mathcal{U}| \ge 1+\sum_{i=1}^N(d_i-1)
\]
for every UPB \(\mathcal U\) in \(\bigotimes_{i=1}^N\mathbb C^{d_i}\). They also showed that this bound is not tight if and only if either \(N=2\) and \(d_1=2\), or \(1+\sum_{i=1}^N(d_i-1)\) is odd and at least one \(d_i\) is even.

\subsection{\(O_N\)-tile decompositions}
\label{subsec:tile-to-upb}

Hypercube decomposition methods have been used in constructions of
UPBs~\cite{bennett1999unextendible,shi2020unextendible,shi2022strongly,he2024strong,zhen2024unextendible,che2022constructing,agrawal2019genuinely}. In particular, You \textit{et al.}~\cite{you2023unextendible} introduced the structure of \(O\)-tiles for the $2$-cube decomposition problem. We now formulate the corresponding \(N\)-dimensional version and prove that it induces UPBs by a construction based on tile-wise Fourier product bases and a global stopper state.

Let
    $\mathcal C=\mathbb Z_{d_1}\times\cdots\times\mathbb Z_{d_N}.$
A \textit{tile} of $\mathcal C$ is a Cartesian product
\[
    t=R_1\times\cdots\times R_N,
    \qquad
    \emptyset\neq R_i\subseteq \mathbb Z_{d_i}.
\]
The sets \(R_i\) do not need to consist of consecutive elements in $\mathbb Z_{d_i}$. The size of the tile is
\(|t|=\prod_{i=1}^N |R_i|\). 

\begin{definition}[\(O_N\)-tile decomposition]
\label{def:ON-tile} Let
 $\mathcal C=\bigsqcup_{j=1}^s t_j$ be a disjoint union of tiles $t_j$, $j=1,2,\ldots,s$ with \(s\ge 3\). It  is called an \(O_N\)-tile decomposition if for every subset
\(J\subseteq\{1,2,\ldots,s\}\) with \(1<|J|<s\), the union
\[
    \bigsqcup_{j\in J}t_j
\]
does not form a tile. \end{definition}

For each tile \(t_j = R_1^{(j)} \times \cdots \times R_N^{(j)}\) in Definition~\ref{def:ON-tile},  we define the subspace supported on the tile $t_j$ as \[ \mathcal H_j := \operatorname{span} \bigl\{ \ket{r_1} \ket{r_2} \cdots \ket{r_N} : (r_1,\ldots,r_N) \in t_j \bigr\}. \] Equivalently, \[ \mathcal H_j = \bigotimes_{i=1}^N \operatorname{span} \{\ket{r} : r \in R_i^{(j)}\}. \] Since the tiles $t_1,\ldots,t_s$ form a partition of $\mathcal C$, the computational basis vectors are partitioned accordingly. Hence the subspaces $\mathcal H_1,\ldots,\mathcal H_s$ are mutually orthogonal and give an orthogonal decomposition $ \mathcal H = \bigoplus_{j=1}^s \mathcal H_j .$ 

Now we describe the product-state construction associated with an \(O_N\)-tile decomposition. 
We first give an example.
\begin{example}[The \(O_3\)-tile decomposition underlying the ``Shifts'' UPB]
\label{ex:shifts-tile}
Consider the three-qubit case
   $ \mathcal C=\mathbb Z_2\times\mathbb Z_2\times\mathbb Z_2.$
The following five tiles form an \(O_3\)-tile decomposition of $\mathcal C$.

\medskip

\noindent
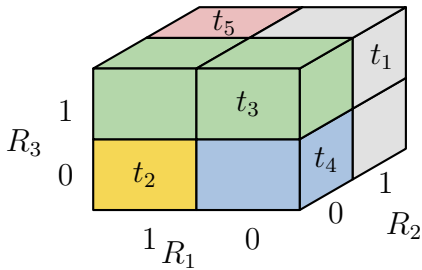
\begin{minipage}[c]{0.55\textwidth}
\centering
\begin{tikzpicture}[
    scale=0.78,
    line join=round,
    line cap=round,
    every node/.style={font=\large}
]

\newcommand{\pt}[3]{({1.75*(#1)+0.90*(#2)},{0.52*(#2)+1.20*(#3)})}

\definecolor{tOne}{RGB}{225,225,225}
\definecolor{tTwo}{RGB}{245,215,85}
\definecolor{tThree}{RGB}{180,215,170}
\definecolor{tFour}{RGB}{170,195,225}
\definecolor{tFive}{RGB}{235,190,190}

\path[draw=black, thick, fill=tFour]
    \pt{2}{0}{0} -- \pt{2}{1}{0} -- \pt{2}{1}{1} -- \pt{2}{0}{1} -- cycle;
\node at \pt{2}{0.5}{0.5} {$t_4$};

\path[draw=black, thick, fill=tOne]
    \pt{2}{1}{0} -- \pt{2}{2}{0} -- \pt{2}{2}{1} -- \pt{2}{1}{1} -- cycle;
\node at \pt{2}{1.55}{0.5} {};

\path[draw=black, thick, fill=tThree]
    \pt{2}{0}{1} -- \pt{2}{1}{1} -- \pt{2}{1}{2} -- \pt{2}{0}{2} -- cycle;
\node at \pt{2}{0.5}{1.5} {};

\path[draw=black, thick, fill=tOne]
    \pt{2}{1}{1} -- \pt{2}{2}{1} -- \pt{2}{2}{2} -- \pt{2}{1}{2} -- cycle;
\node at \pt{2}{1.55}{1.5} {$t_1$};

\path[draw=black, thick, fill=tThree]
    \pt{0}{0}{2} -- \pt{1}{0}{2} -- \pt{1}{1}{2} -- \pt{0}{1}{2} -- cycle;
\node at \pt{0.5}{0.5}{2} {};

\path[draw=black, thick, fill=tThree]
    \pt{1}{0}{2} -- \pt{2}{0}{2} -- \pt{2}{1}{2} -- \pt{1}{1}{2} -- cycle;
\node at \pt{1.5}{0.5}{2} {};

\path[draw=black, thick, fill=tFive]
    \pt{0}{1}{2} -- \pt{1}{1}{2} -- \pt{1}{2}{2} -- \pt{0}{2}{2} -- cycle;
\node at \pt{0.5}{1.5}{2} {$t_5$};

\path[draw=black, thick, fill=tOne]
    \pt{1}{1}{2} -- \pt{2}{1}{2} -- \pt{2}{2}{2} -- \pt{1}{2}{2} -- cycle;
\node at \pt{1.5}{1.5}{2} {};

\path[draw=black, thick, fill=tTwo]
    \pt{0}{0}{0} -- \pt{1}{0}{0} -- \pt{1}{0}{1} -- \pt{0}{0}{1} -- cycle;
\node at \pt{0.5}{0}{0.5} {$t_2$};

\path[draw=black, thick, fill=tFour]
    \pt{1}{0}{0} -- \pt{2}{0}{0} -- \pt{2}{0}{1} -- \pt{1}{0}{1} -- cycle;
\node at \pt{1.5}{0}{0.5} {};

\path[draw=black, thick, fill=tThree]
    \pt{0}{0}{1} -- \pt{1}{0}{1} -- \pt{1}{0}{2} -- \pt{0}{0}{2} -- cycle;
\node at \pt{0.5}{0}{1.5} {};

\path[draw=black, thick, fill=tThree]
    \pt{1}{0}{1} -- \pt{2}{0}{1} -- \pt{2}{0}{2} -- \pt{1}{0}{2} -- cycle;
\node at \pt{1.5}{0}{1.5} {$t_3$};

\node[left] at \pt{-0.4}{0}{0.95} {$R_3$};
\node[left] at \pt{-0.1}{0}{0.5} {$0$};
\node[left] at \pt{-0.1}{0}{1.4} {$1$};

\node[below] at \pt{1.2}{-0.72}{0} {$R_1$};
\node[below] at \pt{0.7}{-0.3}{0} {$1$};
\node[below] at \pt{1.7}{-0.3}{0} {$0$};

\node[right] at \pt{3}{-0.5}{0} {$R_2$};
\node[below right] at \pt{1.9}{0.55}{0} {$0$};
\node[right] at \pt{2.2}{0.9}{0} {$1$};

\end{tikzpicture}
\end{minipage}
\hfill
\begin{minipage}[c]{0.41\textwidth}
\small
\[
\begin{aligned}
    t_1&=\{0\}\times\{1\}\times\{0,1\},\\
    t_2&=\{1\}\times\{0,1\}\times\{0\},\\
    t_3&=\{0,1\}\times\{0\}\times\{1\},\\
    t_4&=\{0\}\times\{0\}\times\{0\},\\
    t_5&=\{1\}\times\{1\}\times\{1\}.
\end{aligned}
\]
\end{minipage}

\medskip

We now illustrate the construction for this decomposition. Omitting normalization
constants, each tile $t_j$ gives a set of tile-wise Fourier product  basis $\mathcal B_j$ as follows. 
\[
\begin{aligned}
\mathcal B_1
&=\{\ket{\eta_1}=|0\rangle|1\rangle(|0\rangle+|1\rangle),\
\ket{\varphi^{(1)}}=|0\rangle|1\rangle(|0\rangle-|1\rangle)\},\\
\mathcal B_2
&=\{\ket{\eta_2}=|1\rangle(|0\rangle+|1\rangle)|0\rangle,\
\ket{\varphi^{(2)}}=|1\rangle(|0\rangle-|1\rangle)|0\rangle\},\\
\mathcal B_3
&=\{\ket{\eta_3}=(|0\rangle+|1\rangle)|0\rangle|1\rangle,\
\ket{\varphi^{(3)}}=(|0\rangle-|1\rangle)|0\rangle|1\rangle\},\\
\mathcal B_4
&=\{\ket{\eta_4}=|0\rangle|0\rangle|0\rangle\},\qquad
\mathcal B_5=\{\ket{\eta_5}=|1\rangle|1\rangle|1\rangle\}.
\end{aligned}
\]
Here, each \(\eta_j\) is the all-zero Fourier vector on the corresponding tile.
The tile-to-UPB construction removes these vectors and adds the global stopper
\[
    |S\rangle=(|0\rangle+|1\rangle)(|0\rangle+|1\rangle)(|0\rangle+|1\rangle).
\]
Thus
\[
\mathcal U
=
\bigcup_{j=1}^5 \bigl(\mathcal B_j\setminus\{\ket{\eta_j}\}\bigr)\cup\{|S\rangle\}
=
\{\ket{\varphi^{(1)}},\ket{\varphi^{(2)}},\ket{\varphi^{(3)}},\ket{S}\},
\]
which is exactly the ``Shifts'' UPB in Example~\ref{shift_upb}.

\end{example}

In general, suppose that there is an \(O_N\)-tile decomposition of $\mathcal C=\mathbb Z_{d_1}\times\cdots\times\mathbb Z_{d_N}$
\[
    \mathcal C=\bigsqcup_{j=1}^s t_j,
    \qquad
    t_j=R_1^{(j)}\times\cdots\times R_N^{(j)}.
\]
Write \(p_i^{(j)}=|R_i^{(j)}|\), and fix an order
\[
    R_i^{(j)}=\{r_{i,0}^{(j)},r_{i,1}^{(j)},\ldots,r_{i,p_i^{(j)}-1}^{(j)}\}.
\]
For \(a\in\mathbb Z_{p_i^{(j)}}\), define the local Fourier state
\begin{equation}
\ket{u_{i,a}^{(j)}}
=
\sum_{\ell=0}^{p_i^{(j)}-1}
\omega_{p_i^{(j)}}^{a\ell}\ket{r_{i,\ell}^{(j)}},
\qquad
\omega_{p_i^{(j)}}=e^{2\pi i/p_i^{(j)}}.
\end{equation}
Then \(\{\ket{u_{i,a}^{(j)}}:a\in\mathbb Z_{p_i^{(j)}}\}\) is an orthogonal basis of \(\operatorname{span}\{\ket{r}:r\in R_i^{(j)}\}\).  Note that the local-zero Fourier state 
$\ket{u_{i,0}^{(j)}}
=\sum_{r\in R_i^{(j)}}\ket{r}.
$

For each tile \(t_j\), let
\begin{equation}
\mathcal B_j=
\left\{
\bigotimes_{i=1}^N \ket{u_{i,a_i}^{(j)}}:
 a_i\in\mathbb Z_{p_i^{(j)}}
\right\}.
\end{equation}
This is an orthogonal product basis of the subspace supported on \(t_j\). Denote the all-zero Fourier state by
\begin{equation}\label{eq:eta}
\ket{\eta_j}:=\bigotimes_{i=1}^N\ket{u_{i,0}^{(j)}}=\bigotimes_{i=1}^N  \left(\sum_{r\in R_i^{(j)}}\ket{r}\right). \end{equation}
Finally, define the  stopper state
\begin{equation}
\ket S=
\bigotimes_{i=1}^N
\left(
\sum_{r\in\mathbb Z_{d_i}}\ket{r}
\right).
\end{equation}
A direct computation gives
\begin{equation}
\left\langle S\middle|
\bigotimes_{i=1}^N \ket{u_{i,a_i}^{(j)}}
\right\rangle\neq0
\quad\Longleftrightarrow\quad
a_1=\cdots=a_N=0.
\end{equation}
Thus \(\ket S\) is orthogonal to every state in \(\mathcal B_j\) except \(\ket{\eta_j}\). We define
\begin{equation}
\mathcal U=
\left(
\bigcup_{j=1}^s(\mathcal B_j\setminus\{\ket{\eta_j}\})
\right)
\cup\{\ket S\}.
\label{eq:tile-upb-set}
\end{equation}

\subsection{The tile-to-UPB theorem}
We now show that the construction associated with an
\(O_N\)-tile decomposition does not merely give an orthogonal product set, but
is in fact unextendible. The key observation is that any additional product
state orthogonal to the constructed set would induce a nontrivial union of
tiles that is itself a tile, which is forbidden by the
definition of an \(O_N\)-tile decomposition.

\begin{theorem}[Tile-to-UPB theorem]
\label{thm:tile-to-UPB}
Let \(\mathcal C=\mathbb Z_{d_1}\times\cdots\times\mathbb Z_{d_N}\), and let
    $\mathcal C=\bigsqcup_{j=1}^s t_j$
be an \(O_N\)-tile decomposition. Then the set \(\mathcal U\) in \eqref{eq:tile-upb-set} is a UPB in \(  \mathbb C^{d_1}\otimes\cdots\otimes\mathbb C^{d_N}.\)
Moreover,
\begin{equation}
    |\mathcal U|=\prod_{i=1}^N d_i-s+1.
    \label{eq:upb_size}
\end{equation}
\end{theorem}

\begin{proof}
First, \(\mathcal U\) is a set of product states. The states within each
\(\mathcal B_j\) form an orthogonal product basis
of the tile-supported subspace \(\mathcal H_j\).
If \(t_j\) and \(t_{j'}\)
are distinct tiles, then they are disjoint Cartesian products; hence for some
coordinate \(i\), the local supports \(R_i^{(j)}\) and \(R_i^{(j')}\) are
disjoint, and the corresponding product states are orthogonal. Finally, the
stopper \(\ket S\) is orthogonal to every nonconstant Fourier state in each
block. Therefore \(\mathcal U\) is a set of mutually orthogonal product states.

It remains to prove unextendibility. Suppose, for contradiction, that there is
a nonzero product state
\[
    \ket{x}=\ket{x_1}\ket{x_2}\cdots\ket{x_N}
\]
orthogonal to every state in \(\mathcal U\).

Let \(\Pi_j\) denote the orthogonal projection onto the tile-supported subspace
\(\mathcal H_j\).
If we write\[
    \ket{x}=\sum_{(r_1,r_2,\ldots,r_N)\in\mathcal C}  c(r_1,r_2,\cdots, r_N) \ket{r_1}\ket{r_2}\cdots\ket{r_N},
\] then \[\Pi_j\ket{x}=\sum_{(r_1,r_2,\ldots,r_N)\in t_j}  c(r_1,r_2,\cdots, r_N) \ket{r_1}\ket{r_2}\cdots\ket{r_N}.\]
For a fixed tile \(t_j\), the
orthogonality of \(\ket{x}\) to all states in
\(\mathcal B_j\setminus\{\ket{\eta_j}\}\) implies that
\(\Pi_j\ket{x}\) is either zero or a scalar multiple of \(\ket{\eta_j}\).
Let
\[
    J=\{j\in[s]:\Pi_j\ket{x}=\alpha_j\ket{\eta_j} \text{ for some }\alpha_j\neq0\}. 
\]
Since \(\ket{\eta_j}\) has nonzero coefficient on every computational basis vector
supported on \(t_j\) and $\Pi_j\ket{x}=0$ for $j\notin J$, we have
\[
c(r_1,\ldots,r_N)\neq 0
\quad\Longleftrightarrow\quad
(r_1,\ldots,r_N)\in \bigcup_{j\in J} t_j .
\]
Hence, defining
\[
\operatorname{supp}(x)
:=
\bigl\{
(r_1,\ldots,r_N)\in\mathcal C
:
c(r_1,\ldots,r_N)\neq 0
\bigr\},
\]
we obtain
\[
\operatorname{supp}(x)
=
\bigcup_{j\in J} t_j .
\]

On the other hand, if we write 
\[
    \ket{x_i}=\sum_{r\in\mathbb Z_{d_i}}c_i(r)\ket r,
\] and define \[
    \operatorname{supp}(x_i)
    =
    \{r\in\mathbb Z_{d_i}:c_i(r)\neq0\},
\] then \[ c(r_1,r_2,\cdots, r_N) = c_1(r_1)c_2(r_2)\cdots c_N(r_N)\] and thus \[
   \operatorname{supp}(x)=
   \operatorname{supp}(x_1)\times\cdots\times \operatorname{supp}(x_N).
\] That is, $\bigcup_{j\in J}t_j$ forms a  Cartesian product.
Since \(\ket{x}\neq0\) and
    $\mathcal H=\bigoplus_{j=1}^s\mathcal H_j,$
the set \(J\) is nonempty. Further,  $\mathcal C=\bigsqcup_{j=1}^s t_j$
is an \(O_N\)-tile decomposition, so $|J|=1$ or $s$.

If  $J=\{j\}$ for some $j$, then $\ket{x}=\Pi_j\ket{x}=\alpha_j\ket{\eta_j}$ for some $\alpha_j\neq 0$. 
But \(\langle S|\eta_j\rangle\neq0\), so
\(\langle S|x\rangle\neq0\), contradicting the assumption that \(\ket{x}\) is
orthogonal to \(\ket S\).

If \(J=[s]\), then $ \operatorname{supp}(x)=\mathcal C$, that is, $c(r_1,r_2,\cdots, r_N)\neq 0$ for any $(r_1,r_2,\cdots, r_N)\in \mathcal C$.
Since 
for each tile \(t_j\),
    $\Pi_j\ket{x}=\alpha_j\ket{\eta_j}$
 and by the form of   \(\ket{\eta_j}\) in \eqref{eq:eta},
$c(r_1,r_2,\cdots, r_N) = c_1(r_1)c_2(r_2)\cdots c_N(r_N)$ is constant on \(t_j\).
Recall $t_j=R_1^{(j)}\times\cdots\times R_N^{(j)}$. We claim that each local coefficient \(c_i\) is constant on each \(R_i^{(j)}\). Indeed, fix a tile
\(t_j\), a coordinate \(i\), and two elements \(r,r'\in R_i^{(j)}\). Fix an
arbitrary \(r_\ell\in R_\ell^{(j)}\) for all \(\ell\neq i\). Since the above
product is constant on \(t_j\), we have
\[
    c_i(r)\prod_{\ell\neq i}c_\ell(r_\ell)
    =
    c_i(r')\prod_{\ell\neq i}c_\ell(r_\ell).
\]
All factors \(c_\ell(r_\ell)\) are nonzero because \(\ket{x}\) has full
support. Hence \(c_i(r)=c_i(r')\). Thus \(c_i\) is constant on \(R_i^{(j)}\).

Next, we show that for each $i$,  \(c_i\) is constant on \(\mathbb Z_{d_i}\).
Suppose that some local coefficient function \(c_i\) is not constant on
\(\mathbb Z_{d_i}\). Then there is a value \(\lambda\in\mathbb C\) such that
\[
    L_i=\{r\in\mathbb Z_{d_i}: c_i(r)=\lambda\}
\]
is nonempty and proper. Consider the Cartesian product
\[
    \Omega_i(L_i)
    =
    \mathbb Z_{d_1}\times\cdots\times L_i\times\cdots\times\mathbb Z_{d_N}.
\]
Since \(c_i(r)\) is constant on each \(R_i^{(j)}\), every \(R_i^{(j)}\) is either
contained in \(L_i\) or disjoint from \(L_i\). Hence every tile is either
entirely contained in \(\Omega_i(L_i)\) or disjoint from it. Therefore
\(\Omega_i(L_i)\) is a nonempty proper union of tiles.
Let
\[
    I=\{j\in[s]:t_j\subseteq \Omega_i(L_i)\}.
\]
Then \(I\) is nonempty and proper. By Definition~\ref{def:ON-tile}, $|I|=1$,
that is, \(\Omega_i(L_i)\) itself is one tile. However, its complement
\[
    \mathcal C\setminus \Omega_i(L_i)
    =
    \mathbb Z_{d_1}\times\cdots\times
    (\mathbb Z_{d_i}\setminus L_i)
    \times\cdots\times\mathbb Z_{d_N}
\]
is also a Cartesian product and is the union of the remaining \(s-1\) tiles.
Since \(s\ge3\), we have \(1<s-1<s\), again contradicting
Definition~\ref{def:ON-tile}.

Therefore every \(c_i\) must be constant on \(\mathbb Z_{d_i}\).
Thus \(\ket{x}\) is proportional to the stopper \(\ket S\). This is impossible
because \(\ket{x}\) is assumed to be orthogonal to \(\ket S\). Therefore no
nonzero product state is orthogonal to all states in \(\mathcal U\), and
\(\mathcal U\) is a UPB.

Finally,
\[
    |\mathcal U|=
    \sum_{j=1}^s(|\mathcal B_j|-1)+1
    =\sum_{j=1}^s(|t_j|-1)+1
    =\prod_{i=1}^N d_i-s+1,
\]
because the tiles partition \(\mathcal C\). This proves the formula
\eqref{eq:upb_size}.
\end{proof}
Theorem~\ref{thm:tile-to-UPB} provides a construction of UPBs in \(N\)-partite systems. Motivated by Problem~\ref{Pro_size}, we propose the following problem.

\begin{problem}
\label{Pro_tile_size}
Given \(  \mathcal C=\mathbb Z_{d_1}\times\cdots\times\mathbb Z_{d_N},\)
determine the integers \(s\) for which \(\mathcal C\) admits an
\(O_N\)-tile decomposition with exactly \(s\) tiles.
\end{problem}

In the next section, we encode the \(O_N\)-tile condition as a SAT instance and solve it with a SAT solver, thereby providing a computational method for Problem~\ref{Pro_tile_size}.

\section{SAT-assisted Search for \(O_N\)-tile Decompositions}\label{sec:sat-search}

By Theorem~\ref{thm:tile-to-UPB}, within our framework, constructing UPBs is reduced to finding \(O_N\)-tile decompositions. This section formulates that search as a SAT problem. We first recall the CNF notation used by SAT solvers and then encode the admissibility, covering, non-overlap, non-combinability, and cardinality constraints.

\subsection{SAT formula in CNF}\label{subsec:sat-cnf}
SAT is a prototypical NP-complete problem and provides a flexible framework
for encoding finite combinatorial search problems. Once an exact encoding has
been established, solving the original search problem amounts to determining
whether the corresponding Boolean formula is satisfiable. Moreover, modern SAT
solvers are often effective on structured instances arising from practical
applications.

Although Boolean formulas can be represented in various forms, modern SAT
solvers typically accept formulas in conjunctive normal form (CNF). The regular
clause structure of CNF is well suited to the propagation techniques and
heuristic optimizations employed by state-of-the-art SAT solvers.

\begin{definition}
A Boolean formula $\mathcal{F}$ is said to be in \textbf{Conjunctive Normal Form } (CNF) if it is constructed as a conjunction of clauses, where each clause consists of a disjunction of literals. Let $X = \{x_1, x_2, \ldots, x_n\}$ denote the set of Boolean variables. A literal $l$ represents either a variable $x_i$ or its negation $\neg x_i$. The formula is structured as follows:$$\mathcal{F} = C_1 \land C_2 \land \cdots \land C_m,$$where each clause $C_i = (l_{i,1} \lor l_{i,2} \lor \cdots \lor l_{i,k_i})$ is a disjunction of literals. The SAT problem seeks to identify a truth assignment $\sigma: X \to \{0, 1\}$ such that the entire formula evaluates to true, \textit{i.e.}, $\mathcal{F}(\sigma) = 1$.
\end{definition}

This CNF formalism is used below to encode the search for \(O_N\)-tile decompositions. By Theorem~\ref{thm:tile-to-UPB}, every satisfying assignment found by the search induces a UPB.

\subsection{Boolean encoding of \(O_N\)-tile decompositions}
\subsubsection{Modeling the \(O_N\)-tile constraints}

We consider decompositions of the \(N\)-cube
    $\mathcal C=\mathbb Z_{d_1}\times\cdots\times\mathbb Z_{d_N}.$
For convenience, we say a tile  \(t=R_1\times\cdots\times R_N\)  of $\mathcal C$  \emph{admissible} if at least two of the sets \(R_i\) are proper subsets of \(\mathbb Z_{d_i}\). Equivalently,
\[
    \left|\{i:R_i\subsetneq \mathbb Z_{d_i}\}\right|\ge 2. \] By Definition~\ref{def:ON-tile}, each tile in an  \(O_N\)-tile decomposition must be admissible (for if a tile is not admissible, its complement is a Cartesian product formed by the union of the remaining $s-1$ tiles, contradicting Definition~\ref{def:ON-tile} since $s\ge 3$).
The Boolean variables range over the admissible tiles. Thus we define
\begin{equation}
\Omega=
\left\{
R_1\times\cdots\times R_N:
\emptyset\neq R_i\subseteq\mathbb Z_{d_i},\quad
\left|\{i:R_i\subsetneq\mathbb Z_{d_i}\}\right|\ge2
\right\}.
\label{eq:admissible-tile-set}
\end{equation}

For a nontrivial Cartesian product
\(T=Q_1\times\cdots\times Q_N\) with \(1<|T|<|\mathcal C|\), and for a point
\(P=(x_1,\ldots,x_N)\in T\), define
\begin{equation}
\Omega_{T,P}
=
\left\{
    t\in\Omega:
    t\subsetneq T\text{ and }P\in t
\right\}.
\label{eq:omega-TP}
\end{equation}
Thus \(\Omega_{T,P}\) consists of the admissible proper subtiles of \(T\) that cover \(P\).

For every \(t\in\Omega\), introduce a Boolean variable \(A_t\), where \(A_t=1\) means that \(t\) is selected in the decomposition. The SAT encoding has the following three groups of constraints.

\begin{enumerate}
\item \textbf{Covering.} Every point of \(\mathcal C\) is covered by at least one selected tile:
\begin{equation}
    \bigwedge_{P\in\mathcal C}
    \bigvee_{\substack{t\in\Omega\\ P\in t}} A_t .
    \label{covering}
\end{equation}

\item \textbf{Non-overlap.} Two distinct selected tiles cannot intersect:
\begin{equation}
    \bigwedge_{\substack{\{t,t'\}\subseteq\Omega\\ t\neq t',\ t\cap t'\neq\emptyset}}
    (\neg A_t\vee\neg A_{t'}).
    \label{non-overlapp}
\end{equation}
Equivalently, one may impose this clause only for an arbitrary fixed ordering \(t<t'\) of \(\Omega\).

\item \textbf{Non-combinability.} No nontrivial Cartesian product \(T\) with \(1<|T|<|\mathcal C|\) may be fully covered by selected admissible proper subtiles of \(T\):
\begin{equation}
    \bigwedge_{\substack{T\text{ is a Cartesian product}\\\text{with}~1<|T|<|\mathcal C|}}
    \neg\left(
        \bigwedge_{P\in T}
        \left(
            \bigvee_{t\in\Omega_{T,P}} A_t
        \right)
    \right).
    \label{Oktile}
\end{equation}
\end{enumerate}

Together, the covering and non-overlap clauses force the selected admissible tiles to form a partition of \(\mathcal C\). The non-combinability clauses enforce the condition in Definition~\ref{def:ON-tile}. Thus the SAT formula encodes exactly the search for \(O_N\)-tile decompositions.

\subsubsection{CNF conversion of the non-combinability constraint}

The covering and non-overlap constraints in \eqref{covering} and
\eqref{non-overlapp} are already in CNF. However, the non-combinability constraint
\eqref{Oktile} contains the nested form
\(\neg(\bigwedge(\bigvee\cdots))\). We convert it to CNF by a Tseitin
encoding.

Fix a nontrivial Cartesian product \(T\) with \(1<|T|<|\mathcal C|\) and a point
\(P\in T\). Introduce an auxiliary variable \(B_{T,P}\), interpreted as
`` \(P\) is covered by at least one selected admissible proper subtile of
\(T\) ''. Formally,
\begin{equation}
    B_{T,P}\leftrightarrow \bigvee_{t\in\Omega_{T,P}} A_t.
    \label{eq:tseitin-B}
\end{equation}
This equivalence is encoded by the CNF clauses
\begin{equation}
    \bigwedge_{t\in\Omega_{T,P}}(\neg A_t\vee B_{T,P})
    \label{eq:tseitin-forward}
\end{equation}
and
\begin{equation}
    \neg B_{T,P}\vee\left(\bigvee_{t\in\Omega_{T,P}} A_t\right).
    \label{eq:tseitin-backward}
\end{equation}
If \(\Omega_{T,P}=\emptyset\), then \eqref{eq:tseitin-backward} reduces to
\(\neg B_{T,P}\), as expected.

For this fixed \(T\), forbidding a full cover by selected proper subtiles is
now expressed as
\begin{equation}
    \neg\left(\bigwedge_{P\in T}B_{T,P}\right),
\end{equation}
which is the single CNF clause
\begin{equation}
    \bigvee_{P\in T}\neg B_{T,P}.
    \label{eq:noncomb-cnf-clause}
\end{equation}
Therefore, the CNF version of \eqref{Oktile} is obtained by imposing \eqref{eq:tseitin-forward} and \eqref{eq:tseitin-backward} for every nontrivial Cartesian product \(T\) and every \(P\in T\), together with \eqref{eq:noncomb-cnf-clause} once for each such \(T\).

\subsubsection{Cardinality constraint and encoding size}

To search for \(O_N\)-tile decompositions with a prescribed number \(s\) of
tiles, we impose the exact-cardinality constraint
\begin{equation}
    \sum_{t\in\Omega}A_t=s.
    \label{eq:exact-cardinality}
\end{equation}
In the implementation, this is encoded as
\begin{equation}
    \sum_{t\in\Omega}A_t\le s
    \label{eq:at-most-s}
\end{equation}
and
\begin{equation}
    \sum_{t\in\Omega} (1 -A_t)\le |\Omega|-s,
    \label{eq:at-least-s}
\end{equation}
where \eqref{eq:at-least-s} is equivalent to
\(\sum_{t\in\Omega}A_t\ge s\). Each at-most constraint is translated into CNF
using a standard totalizer encoding.

We also record the size of the base SAT encoding. Let \( D=|\mathcal C|=\prod_{i=1}^N d_i\). The number of tile-selection variables is
\begin{equation}
    |\Omega|
    =
    \prod_{i=1}^N(2^{d_i}-1)
    -
    \sum_{i=1}^N(2^{d_i}-1)
    +
    (N-1),
    \label{eq:omega-size}
\end{equation}
where the subtracted terms correspond to tiles that are full in all but
possibly one coordinate.

The Tseitin variables \(B_{T,P}\) are indexed by pairs \((T,P)\), where \(T\)
is a nontrivial Cartesian product and \(P\in T\). Their number is
\begin{equation}
    N_B
    =
    D\left(2^{\sum_{i=1}^N d_i-N}-2\right).
    \label{eq:B-size}
\end{equation}
Thus, before adding the auxiliary variables introduced by the totalizer
encoding, the base number of variables is \(N_{\mathrm{var}} = |\Omega|+ N_B.\)

For \(\mathcal C=\mathbb Z_3\times\mathbb Z_3\times\mathbb Z_3\), we have
\(D=27\),
\[
    |\Omega|=7^3-3\cdot7+2=324,
    \qquad
    N_B=27(2^6-2)=1674.
\]
Hence the base SAT encoding contains \(1998\) variables before the cardinality encoding is added.

The SAT constraints above have the intended correctness property.
Given a satisfying assignment, let $  \mathcal T=\{\,t\in\Omega:A_t=1\,\}$ be the set of selected tiles. The \textbf{covering} clauses ensure that every point of
\(\mathcal C\) is covered by at least one tile in \(\mathcal T\), while the
\textbf{non-overlap} clauses ensure that no two selected tiles intersect. Hence
\(\mathcal T\) is a tile partition of \(\mathcal C\). The \textbf{exact-cardinality}
constraint gives \(|\mathcal T|=s\). Finally, the \textbf{non-combinability} clauses
exclude the possibility that the union of any proper subset of selected
tiles with at least two members forms a single Cartesian product. Therefore, any
satisfying assignment of the SAT formula gives an \(O_N\)-tile decomposition of
\(\mathcal C\) with exactly \(s\) tiles. Conversely, any \(O_N\)-tile
decomposition with exactly \(s\) tiles satisfies these constraints by setting
\(A_t=1\) precisely for the tiles in the decomposition and \(A_t=0\) for all
other candidate tiles.

\section{Verification of General UPBs}\label{sec:verification}

This section develops a verification method for general finite sets of product
states. The method is based on the orthogonality-graph representation of UPBs:
it takes as input an arbitrary finite product-state set and tests whether it is
a UPB via local orthogonality graphs and unsaturated subspaces. Although
Theorem~\ref{thm:tile-to-UPB} already proves that every \(O_N\)-tile
decomposition gives rise to a UPB, we further apply this general verifier to
the symbolic product bases generated from the SAT-found \(O_N\)-tile
decompositions, thereby producing exact verification certificates for the
constructed UPBs.

\subsection{Orthogonality graph representation of UPBs}
We first recall the graph-theoretic language used to verify general UPB instances. An undirected simple graph $G = (V, E)$ consists of a vertex set $V$ and an edge set $E$, where edges are unordered pairs of distinct vertices. For any vertex $v$, its neighborhood $N_G(v)$ contains all vertices adjacent to $v$. The complete graph $K_n$ is a graph with $n$ vertices where every pair of distinct vertices is connected. The union of two graphs $G_1=(V_1, E_1)$ and $G_2=(V_2, E_2)$ is defined as $G_1 \cup G_2 = (V_1 \cup V_2, E_1 \cup E_2)$.

We now focus on the orthogonality graph, a concept whose utility in characterizing UPBs was first established by Alon and Lovász~\cite{alon2001unextendible}. Specifically, for a set of nonzero states $\mathcal{V} = \{ \ket{\varphi^{(1)}}, \dots, \ket{\varphi^{(k)}} \}$, the orthogonality graph is defined as $G = (V, E)$ with vertex set $V = \{v_1, \dots, v_k\}$ and edge set 
$$E = \left\{\{v_i, v_j\} \mid i \neq j \text{ and } \braket{\varphi^{(i)} | \varphi^{(j)}} = 0\right\}.$$ Thus, the  orthogonality graph of a UPB is a complete graph.

Let $\mathcal{U} = \left\{ \ket{\varphi^{(i)}} \right\}_{i=1}^k$ be a set of $k$  product states in the $N$-partite Hilbert space $\mathcal{H} = \bigotimes_{m=1}^N \mathbb{C}^{d_m}$; i.e., each state consists of local components: $\ket{\varphi^{(i)}} = \ket{\varphi_1^{(i)}}\ket{\varphi_2^{(i)}}  \cdots  \ket{\varphi_N^{(i)}}$. Based on the local orthogonality relations within each subsystem, we can define $N$ distinct orthogonality graphs.

\begin{definition}[\cite{shi2023graph}]
Let $\mathcal{U}$ be a set of $N$-partite product states in $\mathcal{H} = \bigotimes_{m=1}^N \mathbb{C}^{d_m}$, indexed by $i \in \{1, \dots, k\}$. For each local subsystem $\mathbb{C}^{d_m}$ ($m \in \{1, \dots, N\}$), the local orthogonality graph $G_m = (V, E_m)$ is defined as follows: 
\begin{itemize}
    \item The vertex set $V = \{v_1, v_2, \dots, v_k\}$ corresponds to the set of product states $\mathcal{U} = \left\{ \ket{\varphi^{(i)}} \right\}_{i=1}^k$. 
    \item The edge set $E_m$ contains an edge $(v_i, v_j)$ if and only if the corresponding product states have orthogonal local components in the $m$-th subsystem: 
    $$ E_m = \left\{\{v_i, v_j\} \mid i \neq j,  \braket{\varphi_m^{(i)} | \varphi_m^{(j)}} = 0\right\}.$$
\end{itemize}

For the \(m\)-th subsystem, a set \(W\subseteq V\) is called
\emph{unsaturated} if the corresponding local states do not span the whole
local space, namely
\[
    \operatorname{span}\{\ket{\varphi_m^{(j)}}:v_j\in W\} 
    \neq
    \mathbb C^{d_m}.
\]
\end{definition}

This graph-theoretic framework is essential for the characterization of UPBs. Specifically, Shi \textit{et al.}~\cite{shi2023graph} utilized these orthogonality graphs to establish a necessary and sufficient condition for a set of product states to constitute a UPB. This fundamental criterion serves as the cornerstone of our UPB verification algorithm. Figure~\ref{fig:shifts-orthogonality}
illustrates the local orthogonality graphs for the ``Shifts'' UPB in
Example~\ref{shift_upb}.

\begin{figure}[H]
  \centering
    \tikzset{
    myvertex/.style={circle, fill=black, inner sep=2pt},
        myboundingbox/.style={draw=none, use as bounding box}
  }

    \begin{subfigure}[b]{0.22\textwidth}
    \centering
    \begin{tikzpicture}
      \path[myboundingbox] (-0.5,-0.5) rectangle (2,2);             \node[myvertex, label=left:$v_1$] (v1) at (0,1.5) {};
      \node[myvertex, label=right:$v_2$] (v2) at (1.5,1.5) {};
      \node[myvertex, label=left:$v_3$] (v3) at (0,0) {};
      \node[myvertex, label=right:$v_4$] (v4) at (1.5,0) {};
            \draw [red, thick] (v1) -- (v2);
      \draw [red, thick] (v3) -- (v4);
    \end{tikzpicture}
    \caption{$G_1$}
  \end{subfigure}
  \hfill     \begin{subfigure}[b]{0.22\textwidth}
    \centering
    \begin{tikzpicture}
      \path[myboundingbox] (-0.5,-0.5) rectangle (2,2);
      \node[myvertex, label=left:$v_1$] (v1) at (0,1.5) {};
      \node[myvertex, label=right:$v_2$] (v2) at (1.5,1.5) {};
      \node[myvertex, label=left:$v_3$] (v3) at (0,0) {};
      \node[myvertex, label=right:$v_4$] (v4) at (1.5,0) {};
            \draw [green!60!black, thick] (v1) -- (v3);       \draw [green!60!black, thick] (v2) -- (v4); 
    \end{tikzpicture}
    \caption{$G_2$}
  \end{subfigure}
  \hfill
    \begin{subfigure}[b]{0.22\textwidth}
    \centering
    \begin{tikzpicture}
      \path[myboundingbox] (-0.5,-0.5) rectangle (2,2);
      \node[myvertex, label=left:$v_1$] (v1) at (0,1.5) {};
      \node[myvertex, label=right:$v_2$] (v2) at (1.5,1.5) {};
      \node[myvertex, label=left:$v_3$] (v3) at (0,0) {};
      \node[myvertex, label=right:$v_4$] (v4) at (1.5,0) {};
            \draw [blue, thick] (v1) -- (v4);
      \draw [blue, thick] (v2) -- (v3); 
    \end{tikzpicture}
    \caption{$G_3$}
  \end{subfigure}
  \hfill
    \begin{subfigure}[b]{0.22\textwidth}
    \centering
    \begin{tikzpicture}
      \path[myboundingbox] (-0.5,-0.5) rectangle (2,2);
      \node[myvertex, label=left:$v_1$] (v1) at (0,1.5) {};
      \node[myvertex, label=right:$v_2$] (v2) at (1.5,1.5) {};
      \node[myvertex, label=left:$v_3$] (v3) at (0,0) {};
      \node[myvertex, label=right:$v_4$] (v4) at (1.5,0) {};
            \draw [blue, thick] (v1) -- (v4);
      \draw [blue, thick] (v2) -- (v3);
      \draw [green!60!black, thick] (v1) -- (v3);
      \draw [green!60!black, thick] (v2) -- (v4); 
      \draw [red, thick] (v1) -- (v2);
      \draw [red, thick] (v3) -- (v4);
    \end{tikzpicture}
    \caption{$G_1\cup G_2\cup G_3 = K_4$}   \end{subfigure}
  
  \caption{Orthogonality graph representation of the ``Shifts'' UPB (Example \ref{shift_upb}). The complete graph $K_4$ is decomposed into three edge-disjoint subgraphs $G_1, G_2, G_3$ corresponding to the orthogonality relations in each subsystem.}
  \label{fig:shifts-orthogonality}
\end{figure}
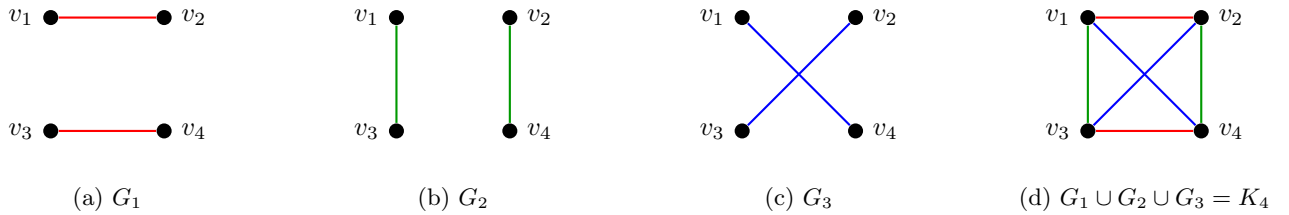

\begin{lemma}[\cite{shi2023graph}]\label{lemma:upb_char}
Let $\mathcal{U}$ be a set of $k$ product states in $\mathcal{H} = \bigotimes_{m=1}^N \mathbb{C}^{d_m}$, and let $\{G_m\}_{m=1}^N$ be the corresponding local orthogonality graphs. The set $\mathcal{U}$ is a UPB if and only if the following two conditions hold:
\begin{enumerate}
    \item \textbf{Orthogonality:} The union of all local graphs forms a complete graph:
    $$ \bigcup_{m=1}^N G_m = K_k. $$
    \item \textbf{Unextendibility:} The vertex set $V$ cannot be covered by a collection of unsaturated sets. That is,
    $$ \bigcup_{m=1}^N W_m \neq V, $$
    for every $N$-tuple $(W_1, \dots, W_N)$ where each \(W_m\) is unsaturated in the \(m\)-th subsystem.
\end{enumerate}
\end{lemma}

\subsection{A subspace verification algorithm for general UPBs}\label{subsec:subspace-verifier}

The orthogonality-graph criterion above gives a verification algorithm for arbitrary finite sets of multipartite product states, not only for those arising from \(O_N\)-tile decompositions. In the present work, Theorem~\ref{thm:tile-to-UPB} already proves unextendibility at the structural level, and the following procedure provides an independent exact verification of the constructed product-state sets. We also benchmark its implementation against QETLAB~\cite{qetlab}. Based on Lemma~\ref{lemma:upb_char}, the verifier has two steps:
\begin{enumerate}
    \item \textbf{Orthogonality verification:} verify that the union of the local
    orthogonality graphs is the complete graph on the constructed product states.
    \item \textbf{Unextendibility verification:} verify that there is no \(N\)-tuple
\((W_1,\ldots,W_N)\), where each \(W_m\) is an unsaturated set for the \(m\)-th
local subsystem, such that \(W_1\cup\cdots\cup W_N = V\).
\end{enumerate}

The second step is accelerated by enumerating only maximal unsaturated sets.
Let \(V=\{1,\ldots,k\}\) be the index set of the constructed product states 

\[
    \ket{\psi^{(j)}}
    =
    \ket{\psi_1^{(j)}}\otimes\cdots\otimes\ket{\psi_N^{(j)}} \in\mathbb C^{d_1}\otimes\cdots\otimes\mathbb C^{d_N},
    \qquad j\in V .
\]
For the \(m\)-th subsystem, a set \(W\subseteq V\) is called  a maximal unsaturated set (MUS) if it is unsaturated and no proper
superset of \(W\) is unsaturated.

We enumerate MUSs by a subspace-closure method. For each subsystem \(m\), we
enumerate all subsets
\[
    S\subseteq V,
    \qquad
    1\le |S|\le d_m-1,
\]
compute
\[
    L_m^{S}
    =
    \operatorname{span}\{\ket{\psi_m^{(j)}}:j\in S\},
\]
and, whenever \(L_m^{S}\neq\mathbb C^{d_m}\), form the closure
\[
    \operatorname{cl}_m(S)
    =
    \{\,j\in V:\ket{\psi_m^{(j)}}\in L_m^{S}\,\}.
\]
The MUSs for the \(m\)-th subsystem are exactly the inclusion-maximal sets among
these closures.

The correctness of this enumeration is immediate from linear algebra. If
\(W\) is an MUS for subsystem \(m\), then
\[
    L=\operatorname{span}\{\ket{\psi_m^{(j)}}:j\in W\}
\]
is a proper subspace of \(\mathbb C^{d_m}\). Hence \(L\) has a basis consisting
of at most \(d_m-1\) states indexed by elements of \(W\). This basis is
enumerated as some subset \(S\subseteq W\), and the closure
\(\operatorname{cl}_m(S)\) recovers \(W\) by maximality. Conversely, every
inclusion-maximal closure is unsaturated by construction and is therefore an
MUS.

Let \(\mathcal L_m\) denote the family of MUSs for the \(m\)-th subsystem. By
Lemma~\ref{lemma:upb_char}, 
the product-state set is unextendible, and hence is certified as a UPB, if and
only if
\[
    W_1\cup\cdots\cup W_N\neq V
\]
for every \((W_1,\ldots,W_N)\in
\mathcal L_1\times\cdots\times\mathcal L_N\).

For fixed local dimensions, the MUS enumeration is polynomial in \(k\). For the
\(m\)-th subsystem, the number of generating subsets is
\[
    \sum_{r=1}^{d_m-1}\binom{k}{r}
    =
    O(k^{d_m-1}).
\]
Using naive linear-algebra routines, computing each span and testing membership
of all \(k\) local states costs \(O(kd_m^3)\). Thus the MUS generation step
for subsystem \(m\) costs \(O(k^{d_m}d_m^3)\), up to the final inclusion-maximal filtering, which can be implemented
efficiently with bitsets.

The pseudocode is given in Appendix~\ref{appendixA}. Table~\ref{Qetlab}
compares the verification times of this method with QETLAB~\cite{qetlab} on representative
instances.~\footnote {All benchmark experiments were performed on the same computer in order to ensure a fair comparison. The machine was equipped with an Intel Core i7-13650HX processor with 14 cores, 16 GB DDR5 memory at 4800 MHz. Our verification algorithm was implemented in Python using exact symbolic arithmetic; for the timing comparison, detailed edge lists and unsaturated-set outputs were disabled. The QETLAB baseline was evaluated by calling the corresponding \textit{IsUPB} function in QETLAB~\cite{qetlab}.}

\begin{table}[htbp]
\centering
\renewcommand{\arraystretch}{1.1}
\small
\resizebox{\textwidth}{!}{\begin{tabular}{|c|c|c|c|c|c|c|c|c|c|c|}
\hline
\multicolumn{2}{|c|}{\textbf{System}} & \multicolumn{3}{c|}{$\mathbb{C}^2 \otimes \mathbb{C}^2 \otimes \mathbb{C}^3$} & \multicolumn{4}{c|}{$\mathbb{C}^2 \otimes \mathbb{C}^2 \otimes \mathbb{C}^4$} & \multicolumn{2}{c|}{$\mathbb{C}^2 \otimes \mathbb{C}^2 \otimes \mathbb{C}^5$} \\ \hline
\multicolumn{2}{|c|}{\textbf{Size}} & 6 & 7 & 8 & 9 & 10 & 11 & 12 & 13 & 14 \\ \hline
\multirow{2}{*}{\textbf{Time (s)}} & QETLAB & 0.0622 & 0.1193 & 0.8317 & 1.4031 & 3.3714 & 11.7792 & 43.0482 & 228.4552 & $>300$ \\ \cline{2-11} 
& Our Algorithm & $0.0495$ & $0.0144$ & $0.0142$ & $0.0216$ & $0.0207$ & $0.0307$ & $0.0457$ & $0.1245$ & $0.2023$ \\ \hline
\end{tabular}}
\caption{Comparison of verification times (in seconds) between QETLAB and our algorithm for UPBs of varying sizes and systems}
\label{Qetlab}
\end{table}

Table~\ref{Qetlab}  indicates that the proposed verifier scales much better than
the QETLAB baseline on the tested instances. While the two methods are
comparable for the smallest example, the running time of QETLAB increases
rapidly as the UPB size grows. In contrast, our algorithm remains below one
second for all listed cases, showing an order-of-magnitude improvement in
verification time. These results demonstrate that the orthogonality-graph and
maximal-unsaturated-set approach provides an efficient exact verification tool
for the product bases considered here.

\section{Results and Discussion}
\subsection{Results}
We demonstrate the effectiveness of our method in constructing and verifying UPBs by solving the $O_3$-tile structure for the 3-cube $\mathcal{C} = \mathbb{Z}_3 \times \mathbb{Z}_3 \times \mathbb{Z}_3$ as a representative example.

\begin{theorem}
There exist UPBs of every size \(k\in\{7,13,14,\ldots,23\}\) in
\(\mathbb C^3\otimes\mathbb C^3\otimes\mathbb C^3\). Among these, sizes
\(7\) and \(19\) are previously known, while the other sizes are obtained by the SAT-assisted \(O_3\)-tile construction of this work.
\end{theorem}

\begin{proof}
The \(k=7\) case follows from the attainability of the
Alon--Lov\'asz lower bound~\cite{alon2001unextendible}. Indeed, \(1+\sum_{i=1}^3(3-1)=7,\) and the bound is attainable because all three local dimensions are odd.

For the remaining cardinalities, the SAT encoding produces
\(O_3\)-tile decompositions of
\(\mathcal C=\mathbb Z_3\times\mathbb Z_3\times\mathbb Z_3\)
with \(s=5,6,\ldots,15\) tiles, as listed in
Appendix~\ref{appendixB}. By Theorem~\ref{thm:tile-to-UPB}, each such
decomposition induces a UPB of cardinality \(27-s+1=28-s\). As \(s\) ranges from \(5\) to \(15\), the resulting cardinalities are
\(23,22,\ldots,13\). Exact verification data for the corresponding
symbolic product-state sets are provided in the supplementary material;
see the Data Availability Statement.
\end{proof}

Table~\ref{result} summarizes UPB sizes obtained from \(O_N\)-tile decompositions found by the SAT search, together with sizes derived by Lemma~\ref{feng} and previously known sizes from the literature. The table notes mark sizes that had already appeared in earlier work. Sizes without notes are, to the best of our knowledge, new instances generated by the present automated \(O_N\)-tile search.~\footnote{For every unmarked entry, the corresponding tile decomposition, symbolic UPB, and verification certificate are provided in the Zenodo archive.}

\begin{lemma}[\cite{feng2006unextendible}]
\label{feng}
If there exists a UPB of size $x$ in $\mathbb{C}^{d_1} \otimes \cdots \otimes \mathbb{C}^{d_N} \otimes \mathbb{C}^{a}$, and a UPB of size $y$ in $\mathbb{C}^{d_1} \otimes \cdots \otimes \mathbb{C}^{d_N} \otimes \mathbb{C}^{b}$, then there exists a UPB of size $x+y$ in $\mathbb{C}^{d_1} \otimes \cdots \otimes \mathbb{C}^{d_N} \otimes \mathbb{C}^{a+b}$.
\end{lemma}

\begin{table}[htbp]
\centering
\begin{threeparttable}
    \caption{Summary of UPB sizes in small tripartite and quadripartite systems}
    \label{result}
    \renewcommand{\arraystretch}{1.2}

    \begin{tabular}{|cc|cc|cc|}
    \hline
    \multicolumn{6}{|c|}{Tripartite systems} \\
    \hline
    system & size & system & size & system & size \\
    \hline

    $\mathbb{C}^2 \otimes \mathbb{C}^2 \otimes \mathbb{C}^2$
    & \textit{4}\tnote{a}
    & $\mathbb{C}^2 \otimes \mathbb{C}^2 \otimes \mathbb{C}^3$
    & $6 \sim 8$\tnote{b}
    & $\mathbb{C}^2 \otimes \mathbb{C}^2 \otimes \mathbb{C}^4$
    & \textit{6}, $8 \sim 12$\tnote{c} \\

    $\mathbb{C}^2 \otimes \mathbb{C}^2 \otimes \mathbb{C}^5$
    & $8 \sim 16$\tnote{d}
    & $\mathbb{C}^2 \otimes \mathbb{C}^2 \otimes \mathbb{C}^6$
    & \textit{8}, $10 \sim 20$\tnote{e}
    & $\mathbb{C}^2 \otimes \mathbb{C}^3 \otimes \mathbb{C}^3$
    & \textit{6}, $9 \sim 14$ \\

    $\mathbb{C}^2 \otimes \mathbb{C}^3 \otimes \mathbb{C}^4$
    & $12 \sim 20$
    & $\mathbb{C}^2 \otimes \mathbb{C}^3 \otimes \mathbb{C}^5$
    & \textit{8}, $15 \sim 26$
    & $\mathbb{C}^2 \otimes \mathbb{C}^3 \otimes \mathbb{C}^6$
    & $18 \sim 32$ \\

    $\mathbb{C}^2 \otimes \mathbb{C}^4 \otimes \mathbb{C}^4$
    & \textit{8}, $16 \sim 28$
    & $\mathbb{C}^2 \otimes \mathbb{C}^4 \otimes \mathbb{C}^5$
    & $18 \sim 36$
    & $\mathbb{C}^2 \otimes \mathbb{C}^4 \otimes \mathbb{C}^6$
    & \textit{10}, $24 \sim 44$ \\

    $\mathbb{C}^3 \otimes \mathbb{C}^3 \otimes \mathbb{C}^3$
    & \textit{7}, $13 \sim 23$\tnote{f}
    & $\mathbb{C}^3 \otimes \mathbb{C}^3 \otimes \mathbb{C}^4$
    & \textit{8}, $17 \sim 32$\tnote{g}
    & $\mathbb{C}^3 \otimes \mathbb{C}^3 \otimes \mathbb{C}^5$
    & \textit{9}, $22 \sim 41$ \\

       \hline
    \multicolumn{6}{|c|}{Quadripartite systems} \\
    \hline
    system & size & system & size & system & size \\
    \hline
    $\mathbb{C}^2 \otimes \mathbb{C}^2 \otimes \mathbb{C}^2 \otimes \mathbb{C}^3$
    & \textit{6}, $11 \sim 20$
    & $\mathbb{C}^2 \otimes \mathbb{C}^2 \otimes \mathbb{C}^3 \otimes \mathbb{C}^3$
    & $16 \sim 32$
    & $\mathbb{C}^2 \otimes \mathbb{C}^3 \otimes \mathbb{C}^3 \otimes \mathbb{C}^3$
    & \textit{8}, $25 \sim 50$ \\
    \hline
    \end{tabular}

    \begin{tablenotes}
        \footnotesize
        \item[a] Ref.~\cite{bennett1999unextendible} established the existence of a UPB of size $4$ in this system.
        \item[b--e] Ref.~\cite{zhang2021new} established the existence of all UPB sizes listed for the corresponding systems.
                                        \item[f] Ref.~\cite{agrawal2019genuinely} established the existence of UPBs of size $19$ in this system.
        \item[g] Ref.~\cite{che2022constructing} established the existence of a UPB of size $19$ in this system.
        \item Italicized entries indicate systems for which the Alon--Lovász bound~\cite{alon2001unextendible} is attainable.

    \end{tablenotes}
\end{threeparttable}
\end{table}

\subsection{Discussion}
\subsubsection{Conclusion}
In this study, we introduced a SAT-assisted framework for constructing UPBs from hypercube decompositions. The main theoretical ingredient is the tile-to-UPB theorem, which shows that every \(O_N\)-tile decomposition induces a UPB through the construction based on tile-wise Fourier product bases and a global stopper state. The SAT formulation turns the search for such decompositions into a SAT problem, and the orthogonality-graph subspace verifier provides exact certificates for the resulting symbolic product bases.

The UPBs obtained here can also serve as seed examples for recursive constructions. As an example, combining Lemma~\ref{feng} from~\cite{feng2006unextendible} with the data in Table~\ref{result}, let
\[ \mathcal A=\{9,10,11,12,13,14\},\qquad \mathcal B=\{13,14,\ldots,23\}. \] 
Then, for any \(d=2x+3y\), where \(x,y\) are nonnegative integers and not both zero, and for any \(a\in\mathcal A\) and \(b\in\mathcal B\), there exists a UPB of size \(ax+by\) in
\(\mathbb C^3\otimes\mathbb C^3\otimes\mathbb C^d\). Thus, the low-dimensional instances produced by the present automated search yield infinite families of multipartite UPBs in larger local dimensions.

Every UPB constructed here yields a completely entangled subspace, namely the
orthogonal complement of its span, and the normalized projector onto this
complement gives a standard UPB-induced PPT entangled state. Thus, the explicit
instances generated by the present method provide additional structured
examples of completely entangled subspaces and bound entangled states.

\subsubsection{Future Work}
Future work may focus on the following directions:
\begin{enumerate}
\item Optimizing the SAT modeling of \(O_N\)-tile decompositions, for example by adding symmetry-breaking constraints, improving the encoding of non-combinability, and developing parallel implementations for larger multipartite systems.

\item Investigating the properties of UPBs induced by \(O_N\)-tile decompositions, with particular attention to their possible cardinalities. Since Theorem~\ref{thm:tile-to-UPB} converts the number of tiles into the size of the resulting UPB, a natural problem is to determine which UPB sizes are attainable by \(O_N\)-tile decompositions in a given multipartite system, and to identify possible obstructions for unattainable sizes.

\item Improving the verification certificates, including more compact certificates for maximal unsaturated sets and exact arithmetic routines for larger local dimensions.

\end{enumerate}

\section*{Data availability statement}

The data supporting the findings of this study consist of the explicit
\(O_N\)-tile decompositions found by the SAT search, the corresponding symbolic product bases obtained from a construction based on tile-wise Fourier product bases and a global stopper state, and the verification certificates generated by the exact orthogonality-graph subspace verifier. These data, together with the source code used for the SAT-based search and the verification procedure, are available at~\url{https://doi.org/10.5281/zenodo.20782480}.

All other data generated or analysed during this study are included in this article and its appendices.

\begin{acknowledgments}
The research of Wanchen Zhang is supported by the Innovation Program for Quantum Science and Technology under Grant No.~2025ZD0300102. Fei Shi acknowledges support from the 2026 Basic Start-up Fund of Sun Yat-sen University under Grant No.~67000-12266020. The research of Xiande Zhang is supported by the National Key Research and Development Program of China under Grant No.~2023YFA1010200, the National Natural Science Foundation of China under Grant Nos.~12171452 and 12231014, and the Quantum Science and Technology--National Science and Technology Major Project under Grant No.~2021ZD0302902.

\end{acknowledgments}

\newpage

\appendix \section{Pseudocode for the verification algorithm}
\label{appendixA}
\begin{algorithm}
\caption{Automated verification of UPBs via local orthogonality graphs}
\label{alg:upb_verification}
\begin{algorithmic}[1]

\Require A set of \(k\) nonzero product vectors
\[
    \mathcal V=\{\ket{\psi^{(j)}}\}_{j=1}^k
    \subseteq \bigotimes_{m=1}^N\mathbb C^{d_m},
    \qquad
    \ket{\psi^{(j)}}
    =
    \bigotimes_{m=1}^N\ket{\psi_m^{(j)}} .
\]

\Ensure \texttt{True} if \(\mathcal V\) is a UPB, and
\texttt{False} otherwise.

\State \(V\gets [k]=\{1,\ldots,k\}\).
\State Construct the local orthogonality graphs
\(G_1,\ldots,G_N\).

\If{\(\displaystyle \bigcup_{m=1}^N G_m\neq K_k\)}
    \State \Return \texttt{False}
\EndIf

\For{each subsystem \(m=1,\ldots,N\)}
    \State \(\mathcal F_m\gets\emptyset\)

    \For{each \(S\subseteq V\) with \(1\le |S|\le d_m-1\)}
        \State Compute
        \[
            L_m(S)
            :=
            \operatorname{span}
            \{\ket{\psi_m^{(j)}}:j\in S\}.
        \]

        \If{\(\dim L_m(S)<d_m\)}
            \State Define
            \[
                \operatorname{cl}_m(S)
                :=
                \left\{
                j\in V:
                \ket{\psi_m^{(j)}}\in L_m(S)
                \right\}.
            \]

            \State
            \[
                \mathcal F_m
                \gets
                \mathcal F_m\cup
                \{\operatorname{cl}_m(S)\}.
            \]
        \EndIf
    \EndFor

    \State Let \(\mathcal M_m\) be the family of inclusion-maximal
    members of \(\mathcal F_m\).
\EndFor

\For{each
\((W_1,\ldots,W_N)\in
\mathcal M_1\times\cdots\times\mathcal M_N\)}
    \If{\(\displaystyle \bigcup_{m=1}^N W_m=V\)}
        \State \Return \texttt{False}
    \EndIf
\EndFor

\State \Return \texttt{True}

\end{algorithmic}
\end{algorithm}

\section{Detailed $O_3$-tile decompositions}
\label{appendixB}

The following lists the tile decompositions of $\mathcal{C} = \mathbb{Z}_3 \times \mathbb{Z}_3 \times \mathbb{Z}_3$ ($|\mathcal{C}|=27$). Each decomposition of size $s$ yields a UPB of size $k = 27 - s + 1$. Tiles are denoted by $t_j = R_1^{(j)} \times R_2^{(j)} \times R_3^{(j)}$.

\noindent\textbf{B.1. $s=5$ Tile Decomposition (Corresponding to UPB Size $23$)}

\noindent $t_1 = \{0, 1, 2\} \times \{0, 2\} \times \{2\}$, \quad $t_2 = \{0, 1, 2\} \times \{1, 2\} \times \{0\}$, \quad $t_3 = \{0, 1, 2\} \times \{0\} \times \{0, 1\}$, 
\newline $t_4 = \{0, 1, 2\} \times \{1\} \times \{1, 2\}$, \quad $t_5 = \{0, 1, 2\} \times \{2\} \times \{1\}$.

\vspace{0.5em}

\noindent\textbf{B.2. $s=6$ Tile Decomposition (Corresponding to UPB Size $22$)}

\noindent $t_1 = \{0, 1, 2\} \times \{0\} \times \{2\}$, \quad $t_2 = \{1, 2\} \times \{0, 2\} \times \{0, 1\}$, \quad $t_3 = \{1, 2\} \times \{1\} \times \{0, 1, 2\}$, 
\newline $t_4 = \{1, 2\} \times \{2\} \times \{2\}$, \quad $t_5 = \{0\} \times \{1, 2\} \times \{0, 1, 2\}$, \quad $t_6 = \{0\} \times \{0\} \times \{0, 1\}$.

\vspace{0.5em}

\noindent\textbf{B.3. $s=7$ Tile Decomposition (Corresponding to UPB Size $21$)}

\noindent $t_1 = \{0, 1, 2\} \times \{0\} \times \{1\}$, \quad $t_2 = \{0, 2\} \times \{0, 1, 2\} \times \{2\}$, \quad $t_3 = \{0, 2\} \times \{0, 1\} \times \{0\}$, 
\newline $t_4 = \{0, 2\} \times \{2\} \times \{0, 1\}$, \quad $t_5 = \{0, 2\} \times \{1\} \times \{1\}$, \quad $t_6 = \{1\} \times \{1, 2\} \times \{0, 1, 2\}$, 
\newline $t_7 = \{1\} \times \{0\} \times \{0, 2\}$.

\vspace{0.5em}

\noindent\textbf{B.4. $s=8$ Tile Decomposition (Corresponding to UPB Size $20$)}

\noindent $t_1 = \{0, 1, 2\} \times \{0, 1\} \times \{2\}$, \quad $t_2 = \{0, 1, 2\} \times \{0\} \times \{0, 1\}$, \quad $t_3 = \{0, 1, 2\} \times \{2\} \times \{1, 2\}$, 
\newline $t_4 = \{0, 1\} \times \{2\} \times \{0\}$, \quad $t_5 = \{0, 2\} \times \{1\} \times \{1\}$, \quad $t_6 = \{2\} \times \{1, 2\} \times \{0\}$, 
\newline $t_7 = \{1\} \times \{1\} \times \{0, 1\}$, \quad $t_8 = \{0\} \times \{1\} \times \{0\}$.

\vspace{0.5em}

\noindent\textbf{B.5. $s=9$ Tile Decomposition (Corresponding to UPB Size $19$)}

\noindent $t_1 = \{0, 1, 2\} \times \{0\} \times \{2\}$, \quad $t_2 = \{0, 1, 2\} \times \{1\} \times \{0\}$, \quad $t_3 = \{0, 2\} \times \{1, 2\} \times \{1\}$, 
\newline $t_4 = \{0, 2\} \times \{0\} \times \{0, 1\}$, \quad $t_5 = \{0, 2\} \times \{2\} \times \{0, 2\}$, \quad $t_6 = \{0, 2\} \times \{1\} \times \{2\}$, 
\newline $t_7 = \{1\} \times \{0, 1, 2\} \times \{1\}$, \quad $t_8 = \{1\} \times \{0, 2\} \times \{0\}$, \quad $t_9 = \{1\} \times \{1, 2\} \times \{2\}$.

\vspace{0.5em}

\noindent\textbf{B.6. $s=10$ Tile Decomposition (Corresponding to UPB Size $18$)}

\noindent $t_1 = \{0, 1, 2\} \times \{1\} \times \{0\}$, \quad $t_2 = \{0, 2\} \times \{0\} \times \{0\}$, \quad $t_3 = \{2\} \times \{0, 1, 2\} \times \{1, 2\}$, 
\newline $t_4 = \{0\} \times \{0, 1\} \times \{1, 2\}$, \quad $t_5 = \{1\} \times \{0, 1\} \times \{2\}$, \quad $t_6 = \{1\} \times \{1, 2\} \times \{1\}$, 
\newline $t_7 = \{0\} \times \{2\} \times \{0, 1, 2\}$, \quad $t_8 = \{1\} \times \{0\} \times \{0, 1\}$, \quad $t_9 = \{1\} \times \{2\} \times \{0, 2\}$, 
\newline $t_{10} = \{2\} \times \{2\} \times \{0\}$.

\vspace{0.5em}

\noindent\textbf{B.7. $s=11$ Tile Decomposition (Corresponding to UPB Size $17$)}

\noindent $t_1 = \{0, 1, 2\} \times \{2\} \times \{1, 2\}$, \quad $t_2 = \{0, 2\} \times \{1\} \times \{2\}$, \quad $t_3 = \{1, 2\} \times \{0\} \times \{2\}$, 
\newline $t_4 = \{1\} \times \{0, 1, 2\} \times \{0\}$, \quad $t_5 = \{2\} \times \{0, 1\} \times \{0, 1\}$, \quad $t_6 = \{0\} \times \{0, 2\} \times \{0\}$, 
\newline $t_7 = \{0\} \times \{0\} \times \{1, 2\}$, \quad $t_8 = \{0\} \times \{1\} \times \{0, 1\}$, \quad $t_9 = \{1\} \times \{1\} \times \{1, 2\}$, 
\newline $t_{10} = \{1\} \times \{0\} \times \{1\}$, \quad $t_{11} = \{2\} \times \{2\} \times \{0\}$.

\vspace{0.5em}

\noindent\textbf{B.8. $s=12$ Tile Decomposition (Corresponding to UPB Size $16$)}

\noindent $t_1 = \{1, 2\} \times \{0, 2\} \times \{2\}$, \quad $t_2 = \{0, 1\} \times \{0\} \times \{1\}$, \quad $t_3 = \{0, 2\} \times \{2\} \times \{0\}$, 
\newline $t_4 = \{1\} \times \{1, 2\} \times \{0, 1\}$, \quad $t_5 = \{2\} \times \{1, 2\} \times \{1\}$, \quad $t_6 = \{0\} \times \{1\} \times \{0, 1, 2\}$, 
\newline $t_7 = \{0\} \times \{0\} \times \{0, 2\}$, \quad $t_8 = \{0\} \times \{2\} \times \{1, 2\}$, \quad $t_9 = \{2\} \times \{0\} \times \{0, 1\}$, 
\newline $t_{10} = \{2\} \times \{1\} \times \{0, 2\}$, \quad $t_{11} = \{1\} \times \{0\} \times \{0\}$, \quad $t_{12} = \{1\} \times \{1\} \times \{2\}$.

\vspace{0.5em}

\noindent\textbf{B.9. $s=13$ Tile Decomposition (Corresponding to UPB Size $15$)}

\noindent $t_1 = \{0, 2\} \times \{1, 2\} \times \{0\}$, \quad $t_2 = \{1, 2\} \times \{0, 1\} \times \{2\}$, \quad $t_3 = \{0, 1\} \times \{1\} \times \{1\}$, 
\newline $t_4 = \{0, 1\} \times \{2\} \times \{2\}$, \quad $t_5 = \{0\} \times \{0, 2\} \times \{1\}$, \quad $t_6 = \{1\} \times \{0, 1\} \times \{0\}$, 
\newline $t_7 = \{2\} \times \{0, 1\} \times \{1\}$, \quad $t_8 = \{0\} \times \{0\} \times \{0, 2\}$, \quad $t_9 = \{1\} \times \{2\} \times \{0, 1\}$, 
\newline $t_{10} = \{2\} \times \{2\} \times \{1, 2\}$, \quad $t_{11} = \{0\} \times \{1\} \times \{2\}$, \quad $t_{12} = \{1\} \times \{0\} \times \{1\}$, 
\newline $t_{13} = \{2\} \times \{0\} \times \{0\}$.

\vspace{0.5em}

\noindent\textbf{B.10. $s=14$ Tile Decomposition (Corresponding to UPB Size $14$)}

\noindent $t_1 = \{0, 1\} \times \{0\} \times \{1\}$, \quad $t_2 = \{0, 1\} \times \{1\} \times \{0\}$, \quad $t_3 = \{0, 1\} \times \{2\} \times \{2\}$, 
\newline $t_4 = \{0, 2\} \times \{1\} \times \{2\}$, \quad $t_5 = \{0, 2\} \times \{2\} \times \{1\}$, \quad $t_6 = \{2\} \times \{0, 1\} \times \{0, 1\}$, 
\newline $t_7 = \{1\} \times \{0, 1\} \times \{2\}$, \quad $t_8 = \{1\} \times \{0, 2\} \times \{0\}$, \quad $t_9 = \{1\} \times \{1, 2\} \times \{1\}$, 
\newline $t_{10} = \{0\} \times \{0\} \times \{0, 2\}$, \quad $t_{11} = \{2\} \times \{2\} \times \{0, 2\}$, \quad $t_{12} = \{0\} \times \{1\} \times \{1\}$, 
\newline $t_{13} = \{0\} \times \{2\} \times \{0\}$, \quad $t_{14} = \{2\} \times \{0\} \times \{2\}$.

\vspace{0.5em}

\noindent\textbf{B.11. $s=15$ Tile Decomposition (Corresponding to UPB Size $13$)}

\noindent $t_1 = \{0, 1\} \times \{2\} \times \{1\}$, \quad $t_2 = \{0, 2\} \times \{1\} \times \{0\}$, \quad $t_3 = \{1, 2\} \times \{1\} \times \{2\}$, 
\newline $t_4 = \{0\} \times \{1, 2\} \times \{2\}$, \quad $t_5 = \{1\} \times \{0, 2\} \times \{0\}$, \quad $t_6 = \{2\} \times \{0, 1\} \times \{1\}$, 
\newline $t_7 = \{2\} \times \{0, 2\} \times \{2\}$, \quad $t_8 = \{0\} \times \{0\} \times \{0, 1, 2\}$, \quad $t_9 = \{1\} \times \{0\} \times \{1, 2\}$, 
\newline $t_{10} = \{1\} \times \{1\} \times \{0, 1\}$, \quad $t_{11} = \{2\} \times \{2\} \times \{0, 1\}$, \quad $t_{12} = \{0\} \times \{1\} \times \{1\}$, 
\newline $t_{13} = \{0\} \times \{2\} \times \{0\}$, \quad $t_{14} = \{1\} \times \{2\} \times \{2\}$, \quad $t_{15} = \{2\} \times \{0\} \times \{0\}$.

\bibliographystyle{unsrt}
\bibliography{references}

\end{document}